\documentclass[preprint,12pt]{elsarticle}
\usepackage{mathrsfs,graphicx,color,hyperref}
\usepackage{amsmath,amsfonts,amssymb,amsthm}

\def \ie {i.e.~}
\def \NP {$\mathcal{NP}$}
\def \bbbz {\mathbb{Z}}

\newtheorem{prop}{Proposition}{\bfseries}{\rmfamily}
\newtheorem{thm}{Theorem}{\bfseries}{\rmfamily}
{\bfseries}{\rmfamily}
\newtheorem{lem}{Lemma}{\bfseries}{\rmfamily}
\newtheorem{cor}{Corollary}{\bfseries}{\rmfamily}
\newtheorem{obs}{Observation}{\bfseries}{\rmfamily}

\newcommand{\ignore}[1]{}

\journal{}

\begin{document}

\begin{frontmatter}

\title{On the additive chromatic number of several families of graphs\tnoteref{grant}}

\author{Daniel Sever\'in\fnref{correspon}}

\address{ Universidad Nacional de Rosario, Argentina \\
          CONICET, Argentina \\
					\texttt{daniel@fceia.unr.edu.ar} }

\tnotetext[grant]{Partially supported by grants PID-ING 416 (UNR), PICT-2013-0586 (MINCYT) and PIP 11220120100277 (CONICET).}
\fntext[correspon]{Corresponding author at Depto. de Matem\'atica, Facultad de Ciencias Exactas y Naturales,
Universidad Nacional de Rosario. \emph{Address}: Pellegrini 250, Rosario, Argentina\\
\emph{E-mail address}: \texttt{daniel@fceia.unr.edu.ar} }

\begin{abstract}
The Additive Coloring Problem is a variation of the Coloring Problem where labels of $\{1,\ldots,k\}$ are assigned to the vertices of a
graph $G$ so that the sum of labels over the neighborhood of each vertex is a proper coloring of $G$.
The least value $k$ for which $G$ admits such labeling is called \emph{additive chromatic number} of $G$.
This problem was first presented by Czerwi\'nski, Grytczuk and \.Zelazny who also proposed a conjecture that for every graph $G$,
the additive chromatic number never exceeds the classic chromatic number.
Up to date, the conjecture has been proved for complete graphs, trees, non-3-colorable planar graphs with girth at least 13 and
non-bipartite planar graphs with girth at least 26.
In this work, we show that the conjecture holds for split graphs.
We also present exact formulas for computing the additive chromatic number for some subfamilies of split graphs
(complete split, headless spiders and complete sun), regular bipartite, complete multipartite, fan, windmill, circuit, wheel,
cycle sun and wheel sun.
\end{abstract}

\begin{keyword}
Additive chromatic number,
Additive coloring conjecture,
Lucky labeling,
Graph algorithms
\end{keyword}

\end{frontmatter}

\section{Introduction}

The \emph{Additive Coloring Problem}, also known as \emph{Lucky Labeling Problem}, was first presented by Czerwi\'nski, Grytczuk and
\.Zelazny in 2009 \cite{LUCKYORIGINAL}. They also proposed the following conjecture:

\medskip
\noindent\textbf{Additive Coloring Conjecture.} \cite{LUCKYORIGINAL} For every graph $G$, $\eta(G) \leq \chi(G)$.
\medskip

Here, $\chi(G)$ is the chromatic number of $G$ and $\eta(G)$ is the additive chromatic number of $G$, defined below.

For a given integer $k$, denote the set $\{1,2,\ldots,k\}$ with $[k]$.
Let $G = (V, E)$ be a finite, undirected and simple graph, and $f : V \rightarrow [k]$ be a labeling of the vertices of $G$.
Denote by $f(S)$ the sum of labels over a set $S \subset V$, \ie $f(S) = \sum_{u \in S} f(u)$.
A labeling is an \emph{additive $k$-coloring} if $f(N(u)) \neq f(N(v))$ for all edges $(u,v) \in E$,
where $N(v)$ is the open neighborhood of $v$.
The \emph{additive chromatic number} of $G$ is defined as the least number $k$ for which $G$ has an
additive $k$-coloring $f$, and is denoted by $\eta(G)$.
The Additive Coloring Problem (ACP) consists of finding such number and is \NP-hard \cite{LUCKYCOMPLEXITY}.

In an attempt to prove the conjecture, several authors gave upper bounds of the additive chromatic number.
For general graphs $G$ with maximum degree $\Delta$, Akbari et al.~proved that $\eta(G) \leq \Delta^2 - \Delta + 1$ \cite{LUCKYCOTASUP}.
For specific families of graphs, we have:
if $G$ is a tree, $\eta(G) \leq 2$ \cite{LUCKYORIGINAL};
if $G$ is planar bipartite, $\eta(G) \leq 3$ \cite{LUCKYORIGINAL};
if $G$ is planar, $\eta(G) \leq 468$ \cite{ADDITIVEPLANAR};
if $G$ is 3-colorable and planar, $\eta(G) \leq 36$ \cite{ADDITIVEPLANAR};
if $G$ is planar with girth at least 13, $\eta(G) \leq 4$ \cite{ADDITIVEPLANAR};
if $G$ is planar with girth at least 26, $\eta(G) \leq 3$ \cite{BRANDT}.

Up to date, the conjecture has been proved for trees, complete graphs (since $\eta(K_n) = n$ \cite{LUCKYORIGINAL}), non-3-colorable
planar graphs with girth at least 13 and non-bipartite planar graphs with girth at least 26.

Our contribution in this work is to give the exact value of the additive chromatic number of several families of graphs and expand the
number of cases in which the conjecture is satisfied.
In addition, we propose a tool which is used for checking the conjecture over all connected graphs up to 10 vertices.

Now, we give some notation used throughout the article.
For each $v \in V$, let $N_G(v)$ be the set of neighbors of $v$ and $d_G(v)$ its degree, \ie $d_G(v) = |N_G(v)|$.
Also, $N_G[v] = N_G(v) \cup \{v\}$.
If $u, v$ are vertices of $G$, we say that $u$ and $v$ are \emph{true twins} if $N_G[u] = N_G[v]$.
Let $D = (V, A)$ be a finite directed graph. For each $v \in V$, define
$N_D(v) = \{(u,v) \in A : u \in V\} \cup \{(v,w) \in A : w \in V\}$. When the graph or digraph is inferred from the context,
we omit the subindex, \ie $d(v), N(v), N[v]$.

Let $D = (V, A)$ be a directed acyclic graph and $G(D)$ be the undirected underlying graph of $D$. We say that $D$
\emph{represents an acyclic orientation} of $G$ if $G(D)$ is isomorphic to $G$.
Let $f: V \rightarrow [k]$ be a labeling of vertices of $D$. If $f(N(u)) < f(N(v))$ for every
$(u, v) \in A$, then $f$ is called \emph{topological additive $k$-numbering} of $D$.
The \emph{topological additive number} of $D$, denoted by $\eta_t(D)$, is defined as the least number $k$ for which $D$ has a topological additive $k$-numbering, or $+\infty$ in case that such $k$ does not exist.
Clearly, $\eta(G) = \min \{ \eta_t(D) ~:~ D~\textrm{represents an acyclic orientation of}~G \}$.

Due to lack of space, we omit the proofs of Propositions. They can be found in an appendix \cite{APPENDIX}
(also at \url{https://arxiv.org/abs/1602.07675}).

\section{Regular bipartite and complete multipartite graphs}

As far as we know, the conjecture has not been proved for general bipartite graphs yet.
We show that the conjecture holds for a subclass of bipartite graphs including regular ones. 

\begin{obs} \label{ADDITIVE1}
Let $G = (V, E)$ be a graph, then $\eta(G) = 1$ if and only if $d(u) \neq d(v)$ for all $(u,v) \in E$.
\end{obs}
\begin{lem}
Let $G = (U \cup V, E)$ be a bipartite graph ($U$ and $V$ are its stable sets) such that,
for all $v \in V$ and $u \in N(v)$, $d(u) < 2 d(v)$. If $d(u) \neq d(v)$ for all $(u,v) \in E$ then $\eta(G) = 1$,
otherwise $\eta(G) = 2$.
\end{lem}
\begin{proof}
In virtue of Observation \ref{ADDITIVE1}, we only have to prove $\eta(G) \leq 2$.
Consider the assignment $f:V \rightarrow \{1,2\}$ such that $f(u) = 2$ for all $u \in U$ and $f(v) = 1$
for all $v \in V$. Then, $f(N(u)) = d(u) < 2 d(v) = f(N(v))$ for all $(u,v) \in E$.
\end{proof}
\begin{cor} \label{REGULARBIP}
If $G$ is a regular bipartite graph, then $\eta(G) = 2$.
\end{cor}

Now, we consider complete multipartite graphs.
We say that a digraph $D$ is \emph{complete $r$-partite} when $G(D)$ is complete $r$-partite.
We say that $D$ is \emph{monotone} when $V(D)$ can be partitioned into subsets $V_1, V_2, \ldots, V_r$ such that
every arc in $V_i \times V_j$ satisfies $i < j$. In order to prove the theorem, we first cite a result given in \cite{IPL2013}:
\begin{lem} \cite{IPL2013} \label{LEMITA}
Let $D$ be a complete $r$-partite digraph. Then, $\eta_t(D) < +\infty$ if and only $D$ is monotone.
In that case,
\[
  \eta_t(D) = \max \biggl\{ \biggl\lceil \dfrac{s_i}{|V_i|} \biggr\rceil : i \in [r] \biggr\},
\]
where $V_1, \ldots, V_r$ is the partition of $V(D)$, $s_r = |V_r|$ and $s_i = \max\{1 + s_{i+1}, |V_i|\}$ for all $i \in [r-1]$.
\end{lem}

\begin{thm}
Let $G=(V_1 \cup \cdots \cup V_r,E)$ be the complete $r$-partite graph ($V_1$,$\ldots$, $V_r$ are its stable sets)
and $|V_i| \geq |V_{i+1}|$ for all $i \in [r-1]$. Then, $\eta(G) = \max \{ \lceil \frac{s_i}{|V_i|} \rceil : i \in [r] \}$ where 
$s_r = |V_r|$ and $s_i = \max\{1 + s_{i+1}, |V_i|\}$ for all $i \in [r-1]$.
Moreover, $\eta(G) \leq r$.
\end{thm}
\begin{proof}
Let $D$ be the monotone digraph such that $G(D) = G$ and the partition of $V(D)$ is $V_1, V_2, \ldots, V_r$.
We must prove that $D$ represents the acyclic orientation of $G$ that provides the lowest value of $\eta_t(D)$.
Let $D'$ be another digraph representing an acyclic orientation of $G$ with $\eta_t(D') < \infty$.
Therefore, $D'$ is a monotone complete $r$-partite digraph
where $G(D')$ is isomorphic to $G$ and the partition of $V(D')$ is $V'_i = V_{\mathfrak{p}(i)}$ for all $i \in [r]$ where
$\mathfrak{p}:[r] \rightarrow [r]$ is some permutation function.
Define $s_i$ and $s'_i$ for $D$ and $D'$ respectively as in Lemma \ref{LEMITA}.
It is easy to verify that sequences $\{s_i\}_{i \in [r]}$ and $\{s'_i\}_{i \in [r]}$ are decreasing.
We now prove that $s'_i \geq s_i$ for all $i$, by induction on $i = r, r-1, \ldots, 1$.
For the case $i = r$ the statement is straightforward.
For $i < r$, suppose that $s'_i < s_i$.
By inductive hypothesis, $1+s_{i+1} \leq s'_i < |V_i|$.
Since $s'_i \geq |V'_j|$ for any $j \geq i$, then $|V_i| > |V'_j|$.
Therefore $\mathfrak{p}$ is not the identity function, and there exists some $k < i$ such that $\mathfrak{p}^{-1}(k) \geq i$,
implying that $|V_i| > |V'_{\mathfrak{p}^{-1}(k)}| = |V_k|$.
This leads to a contradiction as the sets from $\{V_i\}_{i \in [r]}$ are arranged by size in decreasing order.
Thus, $s'_i \geq s_i$ for all $i \in [r]$.

Let $i$ be an integer such that $s_i/|V_i|$ is maximum and $I = \{ t \in [r] : |V_t| = |V_i|\}$. Note that $i$ is the minimum index of $I$.
Let $J = \{ t \in [r] : |V'_t| = |V_i|\}$ and $j$ be the minimum index of $J$. Due to the ordering in the cardinality of sets of $V(D)$,
$i \geq j$. Hence, $s'_j \geq s'_i \geq s_i$.
Since $j \in J$, $|V'_j| = |V_i|$ and we obtain $s'_j/|V'_j| \geq s_i/|V_i|$.
Therefore, $\eta_t(D') \geq \lceil s'_j/|V'_j| \rceil \geq \lceil s_i/|V_i| \rceil = \eta_t(D)$.

Now, we show that $\eta(G) \leq r$. We first prove by induction on $i$ that $s_i \leq |V_i|(r - i + 1)$ for
$i = r, r-1, \ldots, 1$. In first place, if $i = r$, clearly $s_r = |V_r| = |V_r|(r - r + 1)$.
If $i < r$, just two cases are possible.
If $s_i = |V_i|$, clearly $s_i \leq |V_i|(r - i + 1)$.
Otherwise, $s_i = 1 + s_{i+1}$. By the inductive hypothesis $s_{i+1} \leq |V_{i+1}|(r - i)$ and the
fact that $|V_i| \geq |V_{i+1}|$, we obtain:
\[ s_i = 1 + s_{i+1} \leq 1 + |V_{i+1}|(r - i) \leq |V_{i+1}|(r - i + 1) \leq |V_i|(r - i + 1). \]
Hence, $\bigl\lceil \frac{s_i}{|V_i|} \bigr\rceil \leq r - i + 1 \leq r$ for all $i$ and therefore $\eta(G) \leq r$.
\end{proof}

Since $\chi(G) \geq r$ for any complete $r$-partite graph $G$, we conclude that the conjecture holds for these graphs.

\section{Join with complete graphs}

Let $G_1, G_2$ be disjoint graphs. The \emph{join} of $G_1$ with $G_2$, denoted $G_1 \lor G_2$, is defined as the resulting graph $G'$
satisfying $V(G') = V(G_1) \cup V(G_2)$ and $E(G') = E(G_1) \cup E(G_2) \cup \{ (u,v) : u \in V(G_1),~ v \in V(G_2)\}$.
Given a graph $G$, the following theorem allows to solve the ACP of a join of $G$ with a complete graph by just solving the ACP of $G$.

\begin{obs} \label{LOWERBOUNDTWINS}
Let $G = (V, E)$ be a graph and $T \subset V$ such that any $u,v \in T$ are true twins of $G$.
Then, $\eta(G) \geq |T|$.
\end{obs}
\begin{thm} \label{UNIVERTICES}
Let $G$ be a graph of $n$ vertices and $\Delta$ be the largest degree in $G$.
Then, $\eta(G \lor K_q) = \max \{ \eta(G), q \}$ for all $q \leq n-\Delta-1$.
\end{thm}
\begin{proof}
Let $V$ and $E$ be the set of vertices and edges of $G$ respectively, $U = \{ u_1, u_2, \ldots, u_q \}$ be the set of vertices of $K_q$,
$G' = G \lor K_q$ and $f$ be an optimal additive coloring of $G$.
Consider a labeling $f'$ of $G'$ satisfying $f'(v) = f(v)$ for all $v \in V$,
and $f'(u_i) = i$ for all $i \in [q]$. Now, for any $(v,v') \in E$,
$f'(N_{G'}(v)) = f(N_G(v)) + f'(U) \neq f(N_G(v')) + f'(U) = f'(N_{G'}(v'))$.
For any $i,j \in [q]$ such that $i < j$, $f'(N_{G'}(u_i)) = f(U \cup V) - i > f(U \cup V) - j = f'(N_{G'}(u_j))$.
Finally, note that $f'(V \backslash N_G(v)) \geq n - d_G(v) \geq n - \Delta$ for all $v \in V$. Then, for any $u \in U$ and $v \in V$,
$f'(N_{G'}(u)) = f'(U \cup V) - f'(u) \geq f'(U \cup V) - q > f(U \cup V) - n + \Delta \geq f'(U \cup V) - f'(V \backslash N_G(v)) =
f'(N_{G'}(v))$.
Therefore, $f'$ is an additive coloring of $G'$.

In order to prove optimality, note first that any two vertices in $U$ are true twins of $G'$.
By Observation \ref{LOWERBOUNDTWINS}, $\eta(G') \geq q$.
In addition, suppose that $\eta(G') < \eta(G)$. Hence, there exists an additive $k$-coloring $f'$ of $G'$ with $k = \eta(G)-1$.
Let $f$ be the labeling of $G$ satisfying $f(v) = f'(v)$ for all $v \in V$.
We have $f(N_G(v)) = f'(N_{G'}(v)) - f'(U) \neq f'(N_{G'}(v')) - f'(U) = f(N_G(v'))$ for any $(v,v') \in E$.
Therefore, $f$ is an additive $k$-coloring of $G$ which leads to a contradiction.
\end{proof}

When Theorem \ref{UNIVERTICES} is applied one must keep in mind that the size of a complete graph that can be joined to a graph is
limited by $n-\Delta-1$. In fact, if one chooses $q = n-\Delta$, $\eta(G \lor K_q) = \max \{ \eta(G), q \}$ does no longer hold.
For instance, consider the graph $G^*$ of Figure \ref{DEFINITIONS} and $q = 2$. It can be proven that $\eta(G^*)=2$ and $\eta(G^* \lor K_2)=3$.
On the other hand, there are graphs $G$ such that $\eta(G \lor K_q) = \max \{ \eta(G), q \}$ for any $q$. An example is the family of
stable graphs. In that case, $G \lor K_q$ is called complete split. In the next section, we prove that the additive chromatic number of
complete splits is $q$.

The theorem also shows that if the conjecture holds for a graph $G$ then it still holds for $G \lor K_q$ (with
$q \leq n-\Delta-1$) since $\chi(G \lor K_q) = \chi(G) + q$.\\

A vertex $v$ is \emph{universal} in a graph $G$ when $N(v) = V(G) \backslash \{v\}$.
Note that, if $G$ is a graph without universal vertices, the theorem asserts that $\eta(G \lor K_1) = \eta(G)$.
This is the case of fan, windmill and wheel graphs.

Let $n$ be an integer such that $n \geq 3$. A $n$-\emph{fan} is defined as $F_n = P_{n+1} \lor K_1$ where $P_{n+1}$ is a path of length
$n$. Since $\eta(P_{n+1}) = 2$, 
$\eta(F_n) = 2$.

Let $n, m$ be integers such that $n \geq 3$, $m \geq 2$. The \emph{windmill} graph $W_n^m$ is defined as $m$ copies of $K_n$ which
share a single vertex, \ie $W_n^m = m K_{n-1} \lor K_1$. Then, $\eta(W_n^m) = n-1$.

Let $n$ be an integer such that $n \geq 4$. A \emph{wheel} is defined as $W_n = C_n \lor K_1$, where $C_n$ is a circuit of $n$ vertices.
First, we have to known $\eta(C_n)$. If $n$ is even, $C_n$ is regular bipartite. By Corollary \ref{REGULARBIP}, $\eta(W_n) = \eta(C_n) = 2$.
\begin{prop}
Let $n$ be an odd integer such that $n \geq 5$. Then, $\eta(C_n) = 3$.
\end{prop}
Therefore, if $n$ is odd then $\eta(W_n) = 3$.

\section{Split graphs}

A graph $G=(V,E)$ is a \emph{split graph} if $V$ can be partitioned in subsets $Q,S$ such that $Q$ is a clique of $G$ and $S$ is
a stable set of $G$. We denote vertices of $Q$ with $u_1,\ldots,u_q$ and vertices of $S$ with $v_1,\ldots,v_s$.
W.l.o.g. we assume that $Q$ is maximal (unless stated otherwise).
The following result states an upper bound of the additive chromatic number of split graphs.

\begin{thm} \label{COTASUPSPLIT}
Let $G = (Q \cup S, E)$ be a split graph where $Q$ is maximal and $T \subset Q$ be a non-empty set such that the degrees of each vertex
of $T$ differ each other. Then, $\eta(G) \leq |Q|-|T|+1$.
\end{thm}
\begin{proof}
W.l.o.g. let $T = \{u_{q-t+1},u_{q-t+2},\ldots,u_{q-1}, u_q\}$ where $t = |T|$.
We exhibit an additive $(q-t+1)$-coloring of $G$.
Consider the assignment $f:V \rightarrow [q-t+1]$ such that $f(u_i) = i$ for all $i \in [q-t]$, $f(w) = q-t+1$ for all $w \in T \cup S$.
We first check for edges between the clique and the stable set. Let $(u_i, v) \in E$.
Since $Q$ is maximal, for each $v \in S$, there exists $u(v) \in Q$ such that $v$ is not adjacent to $u(v)$.
Then, $f(N(v)) \leq f(Q) - f(u(v)) \leq f(Q) - 1$. On the other hand, let $r_i = |N(u_i) \cap S|$ for all $i \in [q]$.
Since $v \in N(u_i)$, $r_i \geq 1$ and $f(N(u_i)) = f(Q) - f(u_i) + (q-t+1).r_i \geq f(Q)$.
Therefore, $f(N(u_i)) > f(N(v))$.

Now, we check for edges into the clique.
First consider an edge $(u_j, u_k)$ such that $u_j, u_k \in T$. Then, $r_j \neq r_k$ and
$f(N(u_j)) = f(Q) - (q-t+1) + (q-t+1).r_j \neq f(Q) - (q-t+1) + (q-t+1).r_k = f(N(u_k))$.
Finally consider an edge $(u_j, u_k)$ such that $j \in [q-t]$ and $j < k$. Let $\alpha = f(u_k) - f(u_j)$.
Note that $1 \leq \alpha \leq q-t$.
Then, $f(N(u_j)) - f(N(u_k)) = \alpha + (q-t+1).(r_j - r_k)$. Suppose that $(q-t+1).(r_j - r_k) = \alpha$.
Hence, $1 \leq (q-t+1).(r_j - r_k) \leq q-t$. This contradicts $r_j - r_k \in \bbbz$. Therefore, $f(N(u_j)) \neq f(N(u_k))$.
\end{proof}

Observe that $\eta(G) \leq |Q| \leq \chi(G)$, so the conjecture holds for split graphs.
Now, we will see some subfamilies of split graphs where the exact value of $\eta(G)$ can be computed directly.

A \emph{complete split} is a graph $G = (Q' \cup S', E)$ with $|Q'| \geq 1$, $|S'| \geq 2$, $Q'$ is a clique of $G$ and there are
edges $(u, v)$ for all $u \in Q'$ and $v \in S'$. In these graphs, the bound given in Theorem \ref{COTASUPSPLIT} is tight.

\begin{prop}
Let $G = (Q' \cup S', E)$ be a complete split. Then, $\eta(G) = |Q'|$.
\end{prop}

For the next families, we use a result given in \cite{IPL2013}:
\begin{lem} \cite{IPL2013}
Let $D = (V,A)$ be a directed acyclic graph such that its vertices are ordered so that $(u, v) \in A$ implies $u < v$.
If $Q$ is a clique of $G(D)$ and $q_F$, $q_L$ are the smallest and largest vertices of $Q$ respectively, then
\[
  \eta_t(D) \geq \biggl\lceil \dfrac{d(q_F)+1}{d(q_L)-|Q|+2} \biggr\rceil.
\]
\end{lem}
\begin{cor} \label{COTAINF}
Let $G$ be a graph and $Q$ be a clique of $G$. If $d_1$, $d_2$ are the degrees of the vertices of $Q$ with smallest and largest degree
respectively, then
\[
  \eta(G) \geq \biggl\lceil \dfrac{d_1+1}{d_2-|Q|+2} \biggr\rceil.
\]
\end{cor}

A \emph{thin headless spider} of order $q \geq 2$ is a split graph where $|Q|=|S|=q$ and the set of edges between $Q$ and $S$ is
$\{(u_i, v_i) ~:~ i \in [q] \}$.
Figure \ref{DEFINITIONS} shows an example of a thin spider of order 5.

A \emph{thick headless spider} of order $q \geq 2$ is a split graph where $|Q|=|S|=q$ and the set of edges between $Q$ and $S$ is
$\{(u_i, v_j) ~:~ i,j \in [q],~ i \neq j \}$.\\
Equivalently, a thick headless spider is the complement of a thin headless
spider of the same order and vice-versa.

\begin{prop}
Let $G$ be a thin/thick headless spider of order $q$. Then,
$$\eta(G) = \biggl\lceil \dfrac{q+1}{2} \biggr\rceil.$$
\end{prop}

Let $G$ be a graph and $U = \{u_1, \ldots, u_m\} \subset V(G)$. A \emph{sun} is a graph $G'$ obtained from $G$ as follows:
$V(G') = V(G) \cup V$ where $V = \{v_1, \ldots, v_m\}$ and $$E(G') = E(G) \cup \{(u_i, v_{i-1}), (u_i, v_i) ~:~ i \in [m] \}.$$
For the sake of simplicity, $u_0$ and $v_0$ are another names for $u_m$ and $v_m$.
A \emph{complete sun} of order $m$, denoted by $KS_m$, is a split graph obtained from a complete graph $G$ of size $m$.

\begin{prop}
Let $m \geq 3$. Then, $\eta(KS_m) = \bigl\lceil \frac{m+2}{3} \bigr\rceil.$
\end{prop}

\section{Other suns}

In this section, we study \emph{cycle suns} $CS_m$, \ie
when the sun is obtained from a circuit ($V(G) = U$ and $E(G) = \{ (u_i, u_{i-1}) ~:~ i \in [m] \}$),
and \emph{wheel suns} $WS_m$, \ie when the sun is obtained from a wheel ($V(G) = U \cup \{w\}$ and
$E(G) = \{ (u_i, u_{i-1}), (u_i, w) ~:~ i \in [m] \}$).
Figure \ref{DEFINITIONS} displays a wheel graph with $m = 5$.

\begin{prop}
Let $m \geq 4$. Then, $\eta(CS_m) = \eta(WS_m) = 2$.
\end{prop}
Clearly, the conjecture is satisfied in these graphs.

\begin{figure}[ht]
\begin{center}
\includegraphics[scale=0.80]{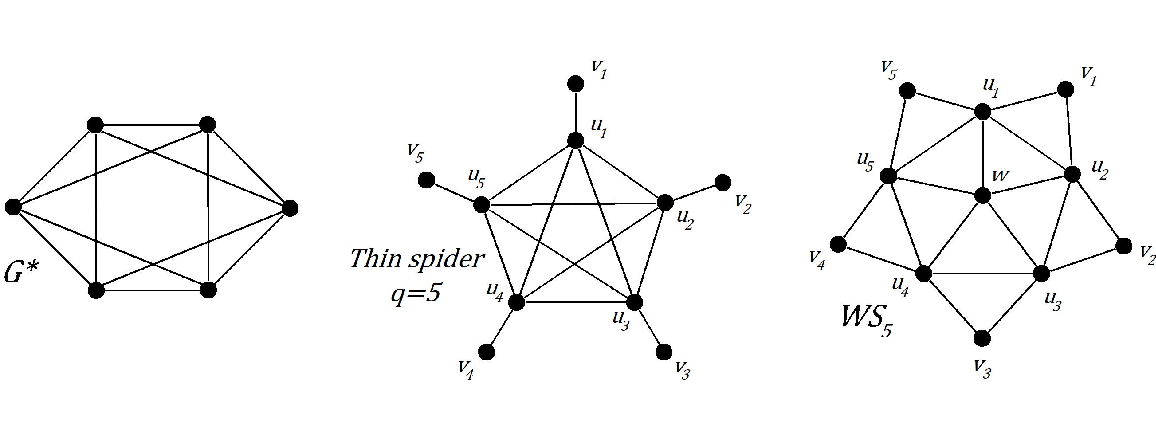}
\end{center}
\vspace{-30pt}
\caption{Some graphs: $G^*$, thin spider and wheel sun of order 5.}
\label{DEFINITIONS}
\end{figure}

\section{A tool for solving the ACP}

As far as we know, there are no tools available for solving ACP. However, we can solve instances of this problem by modeling it
as an integer linear programming formulation and using an available solver like CPLEX.
The source code of such tool can be downloaded from \cite{APPENDIX}.
Besides this tool has been very useful for checking our theoretical results, we have tested the conjecture over all connected graphs up to
10 vertices (about 12 million graphs).

\medskip

\noindent \textbf{Acknowledgements.}
I wish to thank Dr. Graciela Nasini for their helpful comments.


\newpage

\begin{center}
\Large \textbf{Appendix}
\end{center}

\setcounter{section}{0}
\setcounter{prop}{0}
\renewcommand{\thesection}{\Alph{section}}

\section{Proofs of Propositions}

\begin{prop}
Let $n$ be an odd integer such that $n \geq 5$. Then, $\eta(C_n) = 3$.
\end{prop}
\begin{proof}
Let $V = \{v_1, \ldots, v_n\}$ and suppose that $f:V \rightarrow \{1,2\}$ is an additive
2-coloring of $C_n$. Then, $f$ is also a topological additive 2-numbering of a certain digraph $D$ such that $G(D) = C_n$.
Observe that $f(N(v)) \in \{2,3,4\}$ for all $v \in V$.
Since $C_n$ is not bipartite, there must be an oriented path of 3 consecutive vertices in $D$.
W.l.o.g.~assume that $f(N(v_2)) < f(N(v_3)) < f(N(v_4))$. Then, $f(v_1) + f(v_3) = f(N(v_2)) = 2$ and we obtain $f(v_3) = 1$.
But, $f(v_3) + f(v_5) = f(N(v_4)) = 4$ giving $f(v_5) = 3$ which is an absurd. Therefore, $\eta(C_n) \geq 3$.

Consider the assignment $f:V \rightarrow [3]$ such that $f(v_2) = f(v_4) = f(v_5) = 1$, $f(v_1) = 2$, $f(v_3) = 3$
and, if $n \geq 7$, then for $i \geq 6$, $f(v_i) = 1$ if $i$ is even and $f(v_i) = 3$ if $i$ is odd. We obtain
$f(N(v_1)) = 2$ if $n=5$ and $f(N(v_1)) = 4$ otherwise, $f(N(v_2)) = 5$, $f(N(v_3)) = 2$, $f(N(v_4)) = 4$ and
$f(N(v_n)) = 3$. If $n \geq 7$, $f(N(v_5)) = 2$ and $f(N(v_6)) = 4$. If $n \geq 9$, then for $i \in \{7,\ldots,n-1\}$,
$f(N(v_i)) = 2$ if $i$ is odd and $f(N(v_i)) = 6$ if $i$ is even. Thus, $f$ is an additive 3-coloring of $C_n$.  
\end{proof}

\begin{prop}
Let $G = (Q' \cup S', E)$ be a complete split. Then, $\eta(G) = |Q'|$.
\end{prop}
\begin{proof}
Since $G$ has $|Q'|$ true twins, $\eta(G) \geq |Q'|$. On the other hand, let $v \in S'$ and $Q = Q' \cup \{v\}$. Here, $Q$ is
a maximal clique of $G$. By applying Theorem \ref{COTASUPSPLIT}
with $T = \{u, v\}$ where $u \in Q'$, we obtain $\eta(G) = |Q|-1 = |Q'|$.
\end{proof}

\begin{prop}
Let $G$ be a thin/thick headless spider of order $q$. Then,
$$\eta(G) = \biggl\lceil \dfrac{q+1}{2} \biggr\rceil.$$
\end{prop}
\begin{proof}
For the sake of simplicity, we call $r = \lceil \frac{q+1}{2} \rceil$.
We start by proving $\eta(G) = r$ when $G$ is thin.
Note that $d(u_i) = q$ for all $i$. In virtue of Corollary \ref{COTAINF}
we have $\eta(G) \geq r$.
Then, we only need to propose an additive $r$-coloring of $G$.
If $q = 2$, consider the additive 2-coloring $f$ such that $f(u_1) = f(u_2) = f(v_1) = 1$ and $f(v_2) = 2$.
If $q \geq 3$, consider the assignment $f:V \rightarrow [r]$ such that $f(u_i) = r - i + 1$ and $f(v_i) = 1$ for all $i \in [r]$,
and $f(u_i) = q - i + 1$ and $f(v_i) = \lfloor \frac{q+1}{2} \rfloor$ for all $i \in \{r+1,\ldots,q\}$.
We obtain $f(N(u_i)) = f(Q) - f(u_i) + f(v_i) = f(Q) - r + i$ for all $i \in [q]$.
Then, for $j < k$, we have $f(N(u_j)) < f(N(u_k))$.
Regarding the edge $(u_i, v_i)$, we first analyze when $i = 1$.
Note that $f(u_1) = r$, $f(u_2) = r - 1$ and $f(u_q) = 1$, then
$f(N(u_1)) = f(Q) - r + 1 \geq f(u_1) + f(u_2) + f(u_q) - r + 1 = r + 1 > r = f(N(v_1))$.
If $i \geq 2$, $f(N(u_i)) > f(N(u_1)) > f(N(v_1)) = r \geq f(u_i) = f(N(v_i))$.

Now, we consider that $G$ is thick. If $q = 2$ then $G$ is isomorphic to a thin headless spider of order 2. Hence,
assume that $q \geq 3$. Consider the assignment $f:V \rightarrow [r]$ such that $f(u_i) = i$ and $f(v_i) = 1$ for all $i \in [r]$,
and $f(u_i) = r$ and $f(v_i) = i-r+1$ for all $i \in \{r+1,\ldots,q\}$.
We obtain $f(N(u_i)) = f(V) - f(u_i) - f(v_i) = f(V) - i - 1$ for all $i \in [q]$.
Then, for $j < k$, we have $f(N(u_j)) > f(N(u_k))$.
Regarding the edge $(u_j, v_k)$, $j \neq k$, we prove that $f(N(v_j)) \leq f(N(v_1)) < f(N(u_q)) \leq f(N(u_k))$.
As $q \geq k$, the right inequality holds.
The left inequality holds since $f(N(v_j)) = f(Q) - f(u_j) \leq f(Q) - 1 = f(N(v_1))$.
The middle inequality, \ie $f(N(v_1)) < f(N(u_q))$, holds if and only if $f(V) - f(Q) > q$. Observe that
\[ f(V) - f(Q) = \sum_{i=1}^q f(v_i) = r + \sum_{i=r+1}^q (i + 1 - r) = q + \dfrac{(q-r)^2 + q - r}{2} > q. \]
We finish by proving that $\eta(G) \geq r$. Suppose that there exists an additive $(r-1)$-coloring $f$ of $G$.
Recall that $f(N(u_i)) = f(V) - f(u_i) - f(v_i)$ for all $i \in [q]$. Thus, $f(V) - (2r-2) \leq f(N(u_i)) \leq f(V) - 2$. Since there are
$2r-3$ integers in the range of feasible values for $f(N(u_i))$ and $2r-3 < q$, there are two indexes $j$ and $k$ such that
$f(N(u_j))=f(N(u_k))$ by the pigeonhole principle, leading to a contradiction.
\end{proof}

\begin{prop}
Let $m \geq 3$. Then, $\eta(KS_m) = \bigl\lceil \frac{m+2}{3} \bigr\rceil.$
\end{prop}
\begin{proof}
For the sake of simplicity, we call $r = \lceil \frac{m+2}{3} \rceil$.
Note that $d(u_i) = m+1$ for all $i$. In virtue of Corollary \ref{COTAINF}
we have $\eta(G) \geq r$.

We only have to propose an additive $r$-coloring of $KS_m$.
First, define a permutation function $\mathfrak{p}:[m] \rightarrow [m]$ as follows: $\mathfrak{p}(1) = 1$,
$\mathfrak{p}(j) = \frac{j}{2} + 1$ for $j = 2,\ldots,m$ and $j$ even,
$\mathfrak{p}(j) = m - \frac{j-3}{2}$ for $j = 3,\ldots,m$ and $j$ odd.
Clearly, its inverse is: $\mathfrak{q}(1) = 1$,
$\mathfrak{q}(i) = 2(i-1)$ for $i = 2,\ldots,\lfloor \frac{m}{2} \rfloor + 1$,
$\mathfrak{q}(i) = 3 + 2(m-i)$ for $i = \lfloor \frac{m}{2} \rfloor + 2,\ldots,m$.
Let $f$ be the following assignment:
\[ f(u_i) = \begin{cases}
   r,             & m \equiv 2~ (\textrm{mod}~3) ~\land~ i = \mathfrak{p}(m), \\
   \biggl\lfloor \dfrac{\mathfrak{q}(i)}{3} \biggr\rfloor + 1,             & \textrm{otherwise}.
\end{cases} \]
\[ f(v_i) = \begin{cases}
   r + 1 - \biggl\lceil \dfrac{\mathfrak{q}(i)}{3} \biggr\rceil,   & i = 1 ~\lor~ i \geq \biggl\lfloor \dfrac{m}{2} \biggr\rfloor + 2, \\
   2,                                                              & m \equiv 2~ (\textrm{mod}~6) ~\land~ i = \mathfrak{p}(m), \\
   r + 1 - \biggl\lceil \dfrac{\mathfrak{q}(i)+2}{3} \biggr\rceil, & \textrm{otherwise}.
\end{cases} \]
It is easy to check that $f(w) \in [r]$ for all $w \in U \cup V$. Also, observe that first and second case in the definition of $f(v_i)$
do not overlap: if $m \equiv 2~ (\textrm{mod}~6)$, $m$ is even and, therefore,
$2 \leq \mathfrak{p}(m) = m/2 + 1 < \lfloor \frac{m}{2} \rfloor + 2$.

We claim that $f(v_i)$ satisfies the following recursive relationship:
$$f(v_i) = 2r - \mathfrak{q}(i) + f(u_i) - f(v_{i-1}),~~~ \forall~ i \in [m].$$
Then, $f(N(u_i)) = f(U) - f(u_i) + f(v_i) + f(v_{i-1}) = f(U) + 2r - \mathfrak{q}(i)$ for all $i$.
Since $\mathfrak{q}$ is injective, $f(N(u_i)) \neq f(N(u_k))$ for all $i \neq k$.
Regarding edges between $U$ and $V$, note that $f(U) > m$ and for any $v \in V$, $v$ has degree 2, then
$f(N(v)) \leq 2r < f(U) + 2r - m \leq f(U) + 2r - \mathfrak{q}(i) = f(N(u_i))$ for all $i$.
Therefore, $f$ is an additive $r$-coloring of $KS_m$.

Now, we check our claim. If $i \neq \mathfrak{p}(m) = \lceil \frac{m}{2} \rceil + 1$ or $m \not\equiv 2~ (\textrm{mod}~3)$,
then $f(u_i) - \mathfrak{q}(i) = 1 - \lceil \frac{2\mathfrak{q}(i)}{3} \rceil$. That is, we have to check
$f(v_i) = 2r + 1 - \lceil \frac{2\mathfrak{q}(i)}{3} \rceil - f(v_{i-1})$.
In the case that $m \equiv 2~ (\textrm{mod}~3)$ and $i = \mathfrak{p}(m) = \lceil \frac{m}{2} \rceil + 1$,
$f(u_i) - \mathfrak{q}(i) = r - m$ and we have to check $f(v_i) = 3r - m - f(v_{i-1})$.
\begin{enumerate}
\item Case $i = 1$: Since $f(v_0) = f(v_m) = r$,
$f(v_1) = r + 1 - \lceil \frac{1}{3} \rceil = 2r + 1 - \lceil \frac{2}{3} \rceil - r$.
\item Case $i = 2$: Since $f(v_1) = r$,
$f(v_2) = r + 1 - \lceil \frac{4}{3} \rceil = 2r + 1 - \lceil \frac{4}{3} \rceil - r$.
\item Case $i = 3,\ldots,\lfloor \frac{m}{2} \rfloor$ or ``$i = \lfloor \frac{m}{2} \rfloor + 1$ when $m \not\equiv 2~ (\textrm{mod}~6)$'':
First, we prove $1 - \lceil \frac{2(i-1)+2}{3} \rceil = \lceil \frac{2(i-1)}{3} \rceil - \lceil \frac{4(i-1)}{3} \rceil$.
If $i \equiv 1~ (\textrm{mod}~ 3)$, let $h = \frac{i-1}{3}$. Then, $1 - \lceil \frac{2(i-1)+2}{3} \rceil = 1 - 2h - \lceil \frac{2}{3} \rceil
= 2h - 4h = \lceil \frac{2(i-1)}{3} \rceil - \lceil \frac{4(i-1)}{3} \rceil$.
Cases when $i \equiv 0$ or $2~ (\textrm{mod}~ 3)$ are analogous. Since $f(v_{i-1}) = r + 1 - \lceil \frac{2(i-1)}{3} \rceil$,
$f(v_i) = r + 1 - \lceil \frac{2(i-1)+2}{3} \rceil = 2r + 1 - \lceil \frac{4(i-1)}{3} \rceil - r - 1 + \lceil \frac{2(i-1)}{3} \rceil$.
\item Case $i = \lfloor \frac{m}{2} \rfloor + 1$ when $m \equiv 2~ (\textrm{mod}~6)$:
Then, $r = \lceil \frac{m+2}{3} \rceil = \frac{m}{3} + 1$, $\mathfrak{q}(i) = m$,
$\mathfrak{q}(i-1) = m-2$, $f(v_{i-1}) = r + 1 - \lceil\frac{m-2+2}{3}\rceil = 1$ and $f(v_i) = 2 = 3r - m - 1$.
\item Case $i = \lfloor \frac{m}{2} \rfloor + 2$: If $m$ is even, $\mathfrak{q}(i) = m - 1$ and $\mathfrak{q}(i-1) = m$.
If $m \not\equiv 2~ (\textrm{mod}~3)$, $f(v_{i-1}) = r + 1 - \lceil \frac{m+2}{3} \rceil = 1$. Note that
$2r - \lceil \frac{2(m-1)}{3} \rceil = 2$.
If $m \equiv 2~ (\textrm{mod}~3)$, then $m \equiv 2~ (\textrm{mod}~6)$ and, therefore, $f(v_{i-1}) = 2$ and
$2r - \lceil \frac{2(m-1)}{3} \rceil = 3$.
Then, $f(v_i) = r + 1 - \lceil \frac{m - 1}{3} \rceil = 2 = 2r + 1 - \lceil \frac{2(m - 1)}{3} \rceil - f(v_{i-1})$;
If $m$ is odd, $\mathfrak{q}(i) = m$ and $\mathfrak{q}(i-1) = m-1$,
If $m \not\equiv 2~ (\textrm{mod}~3)$, note that $1 - \lceil \frac{m}{3} \rceil = \lceil \frac{m+1}{3} \rceil - \lceil \frac{2m}{3} \rceil$.
Then, $f(v_i) = r + 1 - \lceil \frac{m}{3} \rceil = 2r + 1 - \lceil \frac{2m}{3} \rceil - r - 1 + \lceil \frac{m-1+2}{3} \rceil$.
If $m \equiv 2~ (\textrm{mod}~3)$, $r - 1 = \lceil \frac{m+2}{3} \rceil - 1 = \lceil \frac{m+1}{3} \rceil = \lceil \frac{m}{3} \rceil$ and
$f(v_{i-1}) = r + 1 - \lceil \frac{m-1+2}{3} \rceil = 2$. Then, $f(v_i) = r + 1 - \lceil \frac{m}{3} \rceil = 2 = 3r - m - f(v_{i-1})$.
\item Case $i = \lfloor \frac{m}{2} \rfloor + 3, \ldots, m-1$: We have
$f(v_i) = r + 1 - \lceil \frac{3 + 2(m-i)}{3} \rceil$ and
$2r + 1 - \lceil \frac{2\mathfrak{q}(i)}{3} \rceil - f(v_{i-1}) =
r - \lceil \frac{6 + 4(m-i)}{3} \rceil + \lceil \frac{3 + 2(m-i+1)}{3} \rceil$.
To prove that both expressions are equal, we proceed as in the third case.
\end{enumerate}
\end{proof}

\begin{prop}
Let $m \geq 4$. Then, $\eta(CS_m) = \eta(WS_m) = 2$.
\end{prop}
\begin{proof}
By Observation \ref{ADDITIVE1}
$\eta(CS_m) \geq 2$ and $\eta(WS_m) \geq 2$ so we only have to propose an additive 2-coloring
of $CS_m$ and $WS_m$. We start with $CS_m$.

Consider an assignment $f:V \rightarrow \{1,2\}$ such that $f(u_i) = 2$ if $i$ is odd, $f(u_i) = 1$ if $i$ is even and
$f(v) = 1$ for all $v \in V \backslash \{v_1\}$.
If $m$ is even, also assign $f(v_1) = 1$. Thus, $f(N(u_i)) = 4$ if $i$ is odd, $f(N(u_i)) = 6$ if $i$ is even and
$f(N(v)) = 3$ for all $v \in V$.
If $m$ is odd, assign $f(v_1) = 2$. In this case, $f(N(u_1)) = 6$, $f(N(u_2)) = 7$, $f(N(u_m)) = 5$ and for $i = 3,\ldots,m-1$,
$f(N(u_i)) = 4$ if $i$ is odd and $f(N(u_i)) = 6$ if $i$ is even. In addition, $f(N(v_m)) = 4$ and $f(N(v)) = 3$ for all
$v \in V \backslash \{v_m\}$. Therefore, $f$ is an additive 2-coloring of $CS_m$.

For $WS_m$, assume that $m \neq 5$ and consider the same assignment as before plus $f(w) = 1$. Then,
values of $f(N(v))$ remains the same as in $CS_m$, values of $f(N(u))$ are the same as in $CS_m$ plus one, \ie
$f(N_{WS_m}(u)) = f(N_{CS_m}(u))+1$, and $f(N(w)) = \lceil 3m/2 \rceil$. If $m = 4$, clearly $f$ is an additive 2-coloring of $WS_4$.
If $m \geq 6$, $f(N(w)) > 8 \geq f(N(u))$ and $f$ is an additive 2-coloring of $WS_m$.

For $m = 5$, we propose a different additive 2-coloring of $WS_5$: $f(u_1) = f(u_2) = f(u_4) = f(v_4) = f(v_5) = 1$,
$f(u_3) = f(u_5) = f(v_1) = f(v_2) = f(v_3) = f(w) = 2$. Then, $f(N(v_1)) = 2$, $f(N(v_i)) = 3$ for $i \in \{2,\ldots,5\}$,
$f(N(u_5)) = 6$, $f(N(w)) = 7$, $f(N(u_1)) = f(N(u_3)) = 8$ and $f(N(u_2)) = f(N(u_4)) = 9$.
\end{proof}

\section{Notes on the tool for solving the ACP}

Let $G = (V,E)$ be a graph, $E_2 = \{(u,v), (v,u) : (u,v) \in E\}$ (edges occur in both directions), integer variables $k$ and $f(v)$
for all $v \in V$, and binary variables $z(u,v)$ for all $(u,v) \in E_2$, where $z(u,v) = 1$ if and only if $f(N(u)) < f(N(v))$.
The following integer programming formulation $\mathscr{F}$ computes $\eta(G)$:
\vspace{-35pt}
\begin{center} \begin{align*}
\min k & & \notag \\
\textrm{subject to} & & \\
 & f(N(u)) - f(N(v)) + M_{uv} z(u,v) \leq M_{uv} - 1, & \forall~~(u,v) \in E_2 \\
 & z(u,v) + z(v,u) = 1, & \forall~~(u,v) \in E \\
 & 1 \leq f(v) \leq UB, & \forall~~v \in V \\
 & f(v) \leq k, & \forall~~v \in V \\
 & z(u,v) \in \{0, 1\}, & \forall~~(u,v) \in E_2 \\
 & k, f(v) \in \bbbz_+, & \forall~~v \in V
\end{align*}
\end{center}
where $M_{uv} = 1+|N(u)\backslash N(v)|UB - |N(v)\backslash N(u)|$ for all $(u,v) \in E_2$ and $UB$ is an upper bound of $\eta(G)$.

Additional inequalities can be considered in order to improve the performance of the optimization.
In particular, the initial relaxation of $\mathscr{F}$ can be reinforced by adding these valid inequalities:
\begin{equation*}
z(v,w) + z(w,u) \leq 1,~~~ \textrm{for all}~ u,v,w~ \textrm{such that}~ (u,v) \notin E_2,~ w \in N(u) \subset N(v).
\end{equation*}
Note that if $z(v,w) = z(w,u) = 1$ then $f(N(v)) < f(N(w)) < f(N(u))$, which leads to a contradiction.
We call $\mathscr{F}_1$ to the formulation with these inequalities.

On the other hand, symmetrical solutions arising from the presence of twin vertices can be partially removed with the following
procedure.
Let $\mathscr{C}$ be a partition of $V$, where each element of $\mathscr{C}$ can be: 1) a single vertex, 2) two or more false twins
each other, and 3) two or more true twins each other. Then, for every set of false twins $\{v_1, \ldots, v_t\} \in \mathscr{C}$ add
inequalities $f(v_i) \leq f(v_{i+1}),~~\forall~i \in [t-1]$ and remove variables $z(u,v_i)$, $z(v_i,u)$ and constraints where they occur
for all $i \in 2,\ldots,t$ and $u \in N(v_1)$. Analogously, for every set of true twins $\{v_1, \ldots, v_t\} \in \mathscr{C}$
add inequalities $f(v_i) \leq f(v_{i+1}) - 1,~~\forall~i \in [t-1]$ and remove variables $z(v_i, v_j)$ and
constraints where they occur for all $i,j = 2,\ldots,t$ such that $i \neq j$.
We call $\mathscr{F}_2$ to the resulting formulation after applying this procedure to $\mathscr{F}_1$.

A suitable partition $\mathscr{C}$ can be generated as follows.
First, compute a partition $\mathscr{C'}$ of $V$ into maximal sets of true twins. Let $\mathscr{C}_1 \subset \mathscr{C'}$
composed only of singleton sets and $V' = \bigcup_{W \in \mathscr{C}_1} W$ (\ie $V' = \{ v \in V : \{v\} \in \mathscr{C'} \}$).
Then, compute a partition $\mathscr{C''}$ of $V'$ into maximal sets of false twins.
Finally, do $\mathscr{C} \leftarrow (\mathscr{C'} \backslash \mathscr{C}_1) \cup \mathscr{C''}$.\\

We have run an experiment in order to know the size of graphs where ACP can be solved with our approach.
A computer with an Intel i5 CPU 750@2.67GHz, Visual Studio 2013 and IBM ILOG CPLEX 12.6 has been used for the experiment.
For each instance, the upper bound given by Akbari et al.~(\ie $UB=\Delta^2 - \Delta + 1$) \cite{LUCKYCOTASUP} is computed.
A limit of two hours is imposed to the optimization. In Table \ref{TABLITA}, we show the time in seconds needed to solve 27 random instances
(3 per vertices-density combination) with $\mathscr{F}_1$. These instances were generated by starting from the empty graph of $n$ vertices and adding edges with
probability $p$, where $n \in \{20,25,30\}$ and $p \in \{0.25, 0.5, 0.75\}$ (low, medium and high density, respectively).
A mark ``$-$'' means that the instance could not be solved in the term of two hours.
The last three rows display results over instances of 50 vertices generated by adding 20 true, 20 false or a mix of 10 true and false twins
to the random instances of 30 vertices with medium density. Time is reported in the form $\alpha(\beta)$ where $\alpha$ is the time
consumed by $\mathscr{F}_2$ (including the procedure to generate partition $\mathscr{C}$) and $\beta$ the time consumed by $\mathscr{F}_1$.

As we can see, the tool is able to solve almost all instances of 30 vertices, where harder ones are those with higher density of
edges. In addition, the presence of twin vertices makes instances easier to solve, specially when $\mathscr{F}_2$ is chosen.\\

\begin{table}
\begin{tabular}{|c|c||c@{\hspace{3pt}}c|c@{\hspace{3pt}}c|c@{\hspace{3pt}}c|}
\hline
Vertices & Density & Edges & Time & Edges & Time & Edges & Time \\
\hline
    & Low    &  55 & 0.03 &  40 & 0.02 &  49 & 0.02 \\
 20 & Med. &  88 & 0.09 &  96 & 0.17 &  80 & 0.08 \\
    & High   & 125 & 14.1 & 145 & 2.48 & 138 & 3.47 \\
\hline
    & Low    &  67 & 0.03 &  83 & 0.11 &  76 & 0.03 \\
 25 & Med. & 161 &  332 & 148 & 8.32 & 168 &  896 \\
    & High   & 237 & 54.3 & 216 &  198 & 231 & 7.09 \\
\hline
    & Low    & 130 & 0.20 & 113 & 0.05 & 104 & 1.31 \\
 30 & Med. & 223 & 1696 & 213 &  666 & 219 & 53.5 \\
    & High   & 327 &  $-$ & 316 &  $-$ & 313 &  733 \\
\hline
 30+20f      & & 503 & 0.12(0.23) & 573 & 0.23(0.39) & 559 & 0.19(0.44) \\
 30+10t+10f  & & 668 & 0.61(0.84) & 738 & 1.06(6.45) & 724 & 0.52(12.4) \\
 30+20t      & & 713 & 1.56(4.92) & 783 & 2.38(7.05) & 769 & 3.78(5.63) \\
\hline
\end{tabular}
\caption{Time in seconds needed to solve random instances}
 \label{TABLITA}
\end{table}

Besides this tool has been very useful for checking our theoretical results, we have tested the conjecture over all connected graphs up to
10 vertices (about 12 million graphs). These instances are provided by Brendan McKay:
\begin{center}
\url{http://users.cecs.anu.edu.au/~bdm/data/graphs.html}
\end{center}
while a DSATUR code by Rhyd Lewis have been used for obtaining $\chi(G)$:
\begin{center}
\url{http://rhydlewis.eu/resources/gCol.zip}
\end{center}

For each instance, we assign its chromatic number to $UB$ and solve $\mathscr{F}_2$ (indeed, it is not necessary to reach optimality;
the solver is interrupted when a feasible solution is found). The test fails if an infeasible model is reached, which means that a counterexample to the conjecture is found.
Fortunately, it finished with success (\ie the conjecture is valid for the set of graphs tested).
In particular, all graphs of 9 vertices (261080) were solved in
285 seconds and all graphs of 10 vertices (11716571) were solved in 25027 seconds. The average of time elapsed for each instance of
9 and 10 vertices is 1.09 and 2.14 milliseconds respectively.\\   

\noindent \textbf{Closed additive colorings.}
Because of Additive Coloring Conjecture, it is natural to wonder if a similar conjecture can be proposed by considering a local
identification problem where closed neighborhoods are used instead of open ones. More specifically, consider a graph $G$ without true twins
and call \emph{closed additive $k$-coloring} to a labeling $f : V \rightarrow [k]$ such that $f(N[u]) \neq f(N[v])$ for all edges
$(u,v) \in E$. Then, denote with $\eta[G]$ the least number $k$ for which $G$ has a closed additive $k$-coloring.
We wonder whether the inequality $\eta[G] \leq \chi(G)$ holds for all graph $G$ without true twins, since there are several cases where
this happens (for instance, $\eta[G] = 1$ if and only if $\eta(G) = 1$, then $\eta[G] \leq \chi(G)$ for graphs $G$ such that $\eta[G] = 1$).
We have seen empirically that the inequality always holds for graphs up to 10 vertices.
Unfortunately, we found a family of counterexamples such that $\eta[G] - \chi(G)$ can be as large as one wishes.
Indeed, let $K_n$ be a complete graph of $n$ vertices, with $n \geq 4$. Construct a graph $G$ by replacing every edge $(u,v)$ of $K_n$ by 
a path $\{u, w^1_{uv}, w^2_{uv}, v\}$. The assignment $f(v) = 3$ for all $v \in V(K_n)$, and $f(w^i_{uv}) = i$ for all $(u,v) \in E(K_n)$
is a 3-coloring of $G$. In fact, $\chi(G) = 3$ since $G$ contains a $C_9$.
Now, if $f'$ is a closed additive $k$-coloring of $G$, then $f'(N[w^1_{uv}]) \neq f'(N[w^2_{uv}])$ for all
$(u,v) \in E(K_n)$. Hence, $f'(u) \neq f'(v)$ and each vertex from $V(K_n)$ must have a different labeling. Therefore, $\eta[G] \geq n$.
Figure \ref{COUNTEREXAMPLE} displays the smallest counterexample generated with this construction, which has 16 vertices.

\begin{figure}[ht]
\begin{center}
\includegraphics[scale=0.10]{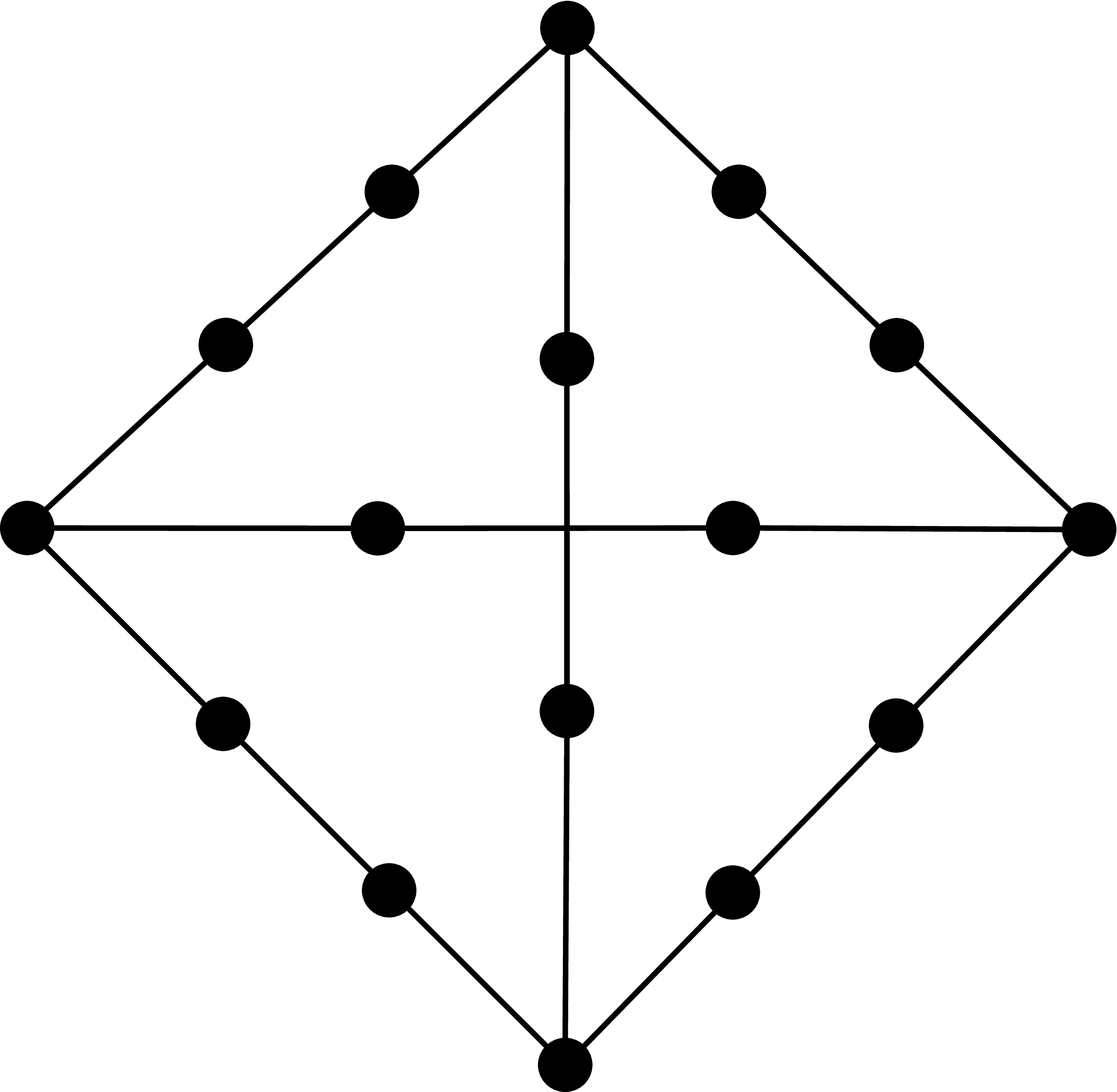}
\end{center}
\vspace{-15pt}
\caption{Graph $G$ such that $\eta[G] = 4 > 3 = \chi(G)$}
\label{COUNTEREXAMPLE}
\end{figure}

A formulation $\mathscr{F}'$ that computes $\eta[G]$ is obtained by replacing constraints $f(N(u)) - f(N(v)) + M_{uv} z(u,v) \leq M_{uv} - 1$ in
$\mathscr{F}$ by $f(N[u]) - f(N[v]) + M_{uv} z(u,v) \leq M_{uv} - 1$, where $M_{uv} = 1+|N[u]\backslash N[v]|UB - |N[v]\backslash N[u]|$
for all $(u,v) \in E_2$. We have also run the experiment with formulation $\mathscr{F}'$ in order to check if $\eta[G] \leq \chi(G)$ for
all connected graphs up to 10 vertices and without true twins. Again, the test finished with success although it took more time, probably
because no additional valid inequalities have been added to $\mathscr{F}'$ (as in the case of $\mathscr{F}$ and $\mathscr{F}_2$).
In particular, all graphs of 9 vertices (197772) were
solved in 724 seconds and all graphs of 10 vertices (9721362) were solved in 46062 seconds. The average of time elapsed for each
instance of 9 and 10 vertices is 3.66 and 4.74 milliseconds respectively.   

About this parameter, $\eta[G]$, M. Axenovich, J. Harant, J. Przybyło, R. Soták, M. Voigt and J. Weidelich published \emph{A note on adjacent vertex distinguishing colorings of graphs} (Discr. Appl. Math. \textbf{205} (2016) 1--7) where they show that
$\eta[G] \leq \Delta^2 - \Delta + 1$ (the same bound of Akbari et al.~for $\eta(G)$ \cite{LUCKYCOTASUP}) and another family of
counterexamples for the conjecture among other results concerning this parameter.
They also extend the definition for the cases where $G$ has true twins.

\section{Examples of graphs and its optimal additive colorings}

Figure \ref{MOREDEF} presents examples of a 4-fan, a windmill graph, a wheel, a complete sun, a cycle sun and a headless thick spider
of order 5. In the drawing of the thick spider, the edges $(u_j,v_k)$ with
$j \neq k$ have been removed for the sake of clarity.
Instead, a dashed line connecting $u_j$ and $v_j$ have been added to remember that these vertices are not connected.

Also additive colorings are shown on Figure \ref{EXAMPLES}:
on the left column, the labelings of an optimal additive coloring are displayed, while on the right,
the values of $f(N(v))$ for each vertex $v$ are reported.

\begin{figure}[ht]
\begin{center}
\includegraphics[scale=0.75]{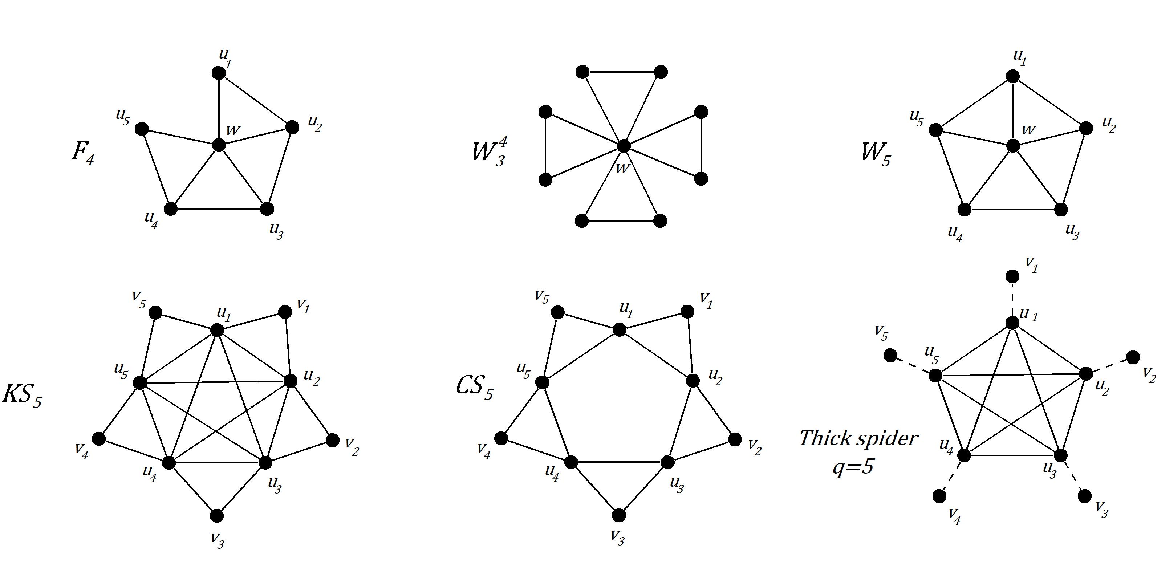}
\end{center}
\caption{Examples of graphs}
\label{MOREDEF}
\end{figure}

\begin{figure}[ht]
\begin{center}
\includegraphics[scale=0.75]{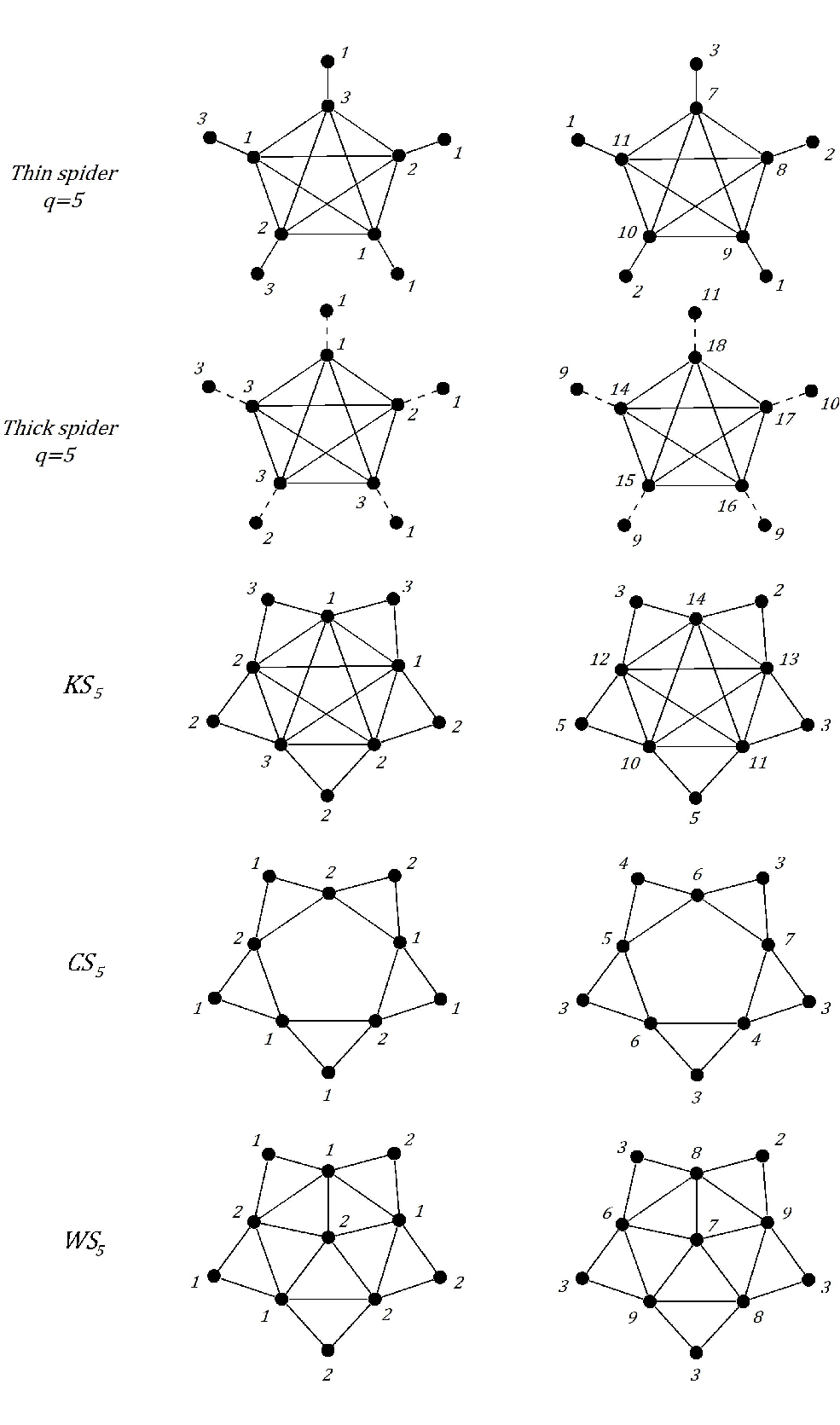}
\end{center}
\caption{Additive colorings of spiders and suns}
\label{EXAMPLES}
\end{figure}

\end{document}